\documentclass[a4paper, 10pt, conference, compsoc]{IEEEtran}
\IEEEoverridecommandlockouts
\usepackage{cite}
\usepackage{amsmath,amssymb,amsfonts,amsthm}
\usepackage{algorithmic}
\usepackage{graphicx}
\usepackage{textcomp}
\usepackage{xcolor}
\def\BibTeX{{\rm B\kern-.05em{\sc i\kern-.025em b}\kern-.08em
    T\kern-.1667em\lower.7ex\hbox{E}\kern-.125emX}}

\newtheorem{theorem}{Theorem}[section]
\newtheorem{lemma}[theorem]{Lemma}

\newtheorem{definition}[theorem]{Definition}

\numberwithin{equation}{section}

\DeclareMathOperator{\Pois}{Poisson}

\begin{document}

\title{Committee Selection is More Similar Than You Think: \\ Evidence from Avalanche and Stellar}

\author{\IEEEauthorblockN{Tarun Chitra}
\IEEEauthorblockA{\textit{Gauntlet Networks, Inc.} \\
tarun@gauntlet.network}
\and
\IEEEauthorblockN{Uthsav Chitra}
\IEEEauthorblockA{\textit{Department of Computer Science} \\
\textit{Princeton University}\\
uchitra@cs.princeton.edu}
}

\maketitle

\begin{abstract}
Increased interest in scalable and high-throughput blockchains has led to an explosion in the number of committee selection methods in the literature. Committee selection mechanisms allow consensus protocols to safely select a \emph{committee}, or a small subset of validators that is permitted to vote and verify a block of transactions, in a distributed ledger. There are many such mechanisms, each with substantially different methodologies and guarantees on communication complexity, resource usage, and fairness. In this paper, we illustrate that, despite these implementation-level differences, there are strong statistical similarities between committee selection mechanisms. We concretely show this by proving that the committee selection of the Avalanche consensus protocol can be used to choose committees in the Stellar Consensus Protocol that satisfy the necessary and sufficient conditions for Byzantine agreement. We also verify these claims using simulations and numerically observe sharp phase transitions as a function of protocol parameters. Our results suggest the existence of a `statistical taxonomy' of committee selection mechanisms in distributed consensus algorithms.
\end{abstract}


\section{Introduction}
Byzantine Fault Tolerant (BFT) algorithms provide a robust, resilient way for a group of $n$ networked participants to come to agreement on an ordered series of transactions, known as a ledger. Although existing algorithms, such as PBFT \cite{castro1999practical} or Paxos \cite{lamport1998part}, provide strong safety guarantees for transactions, they do this at the expense of performance (e.g. transaction rate), permissioning (the identities of all nodes must be known at all times) and synchrony assumptions. In recent years, there has been a burst of interest in distributed consensus algorithms that provide weaker safety guarantees but can achieve higher transaction rates and reduced bandwidth usage. These algorithms, inspired by the empirical success of the Nakamoto consensus mechanism used in Bitcoin, combine advances in cryptography and distributed systems to achieve high performance and low bandwidth usage while providing tail-bounds on the probability of a successful attack (see \cite{garay2015bitcoin, pass2017analysis, fitzi2018}). Such protocols (including Bitcoin) often reduce the BFT requirement, which purports that an algorithm is resistant to arbitrary Byzantine attackers, to specific assumptions about the ways that bad actors can misbehave. Given the low transaction rates of live networks, such as Bitcoin, Ethereum, and Tezos, there has been an uptick in research that aims to improve the scalability of blockchain protocols so that they can be used as a general compute platform \cite{wood2014ethereum, hanke2018dfinity, sompolinsky2015secure}.

One of the key improvements common to most blockchain scaling solutions is the \emph{committee selection mechanism}. Committee selection mechanisms allow users to coordinate and choose a small committee of users without knowing the identities of all other participants. The committee is then allowed to append transactions to the ledger while collecting transaction fees and/or block mining rewards. There are three major camps of committee selection, each with substantially different methodologies and guarantees: the first camp uses selection to choose a small set of validators who achieve consensus via a standard Byzantine Fault Tolerance (BFT) algorithm; the second uses selection to assign validators / miners to specific chains; and the last allows users explicitly decide the validators that they trust. 

Distinct committee selection mechanisms can provide dramatically different guarantees in terms of ease of implementation, fairness, communication complexity, resource usage, and number of rounds of multi-party computation. On one side, social committee selection mechanisms such as Ripple \cite{schwartz2014ripple}, where individual nodes choose which nodes they want to reach consensus with, provide minimal fairness\footnote{In this context, fairness refers to a node receiving $p\pm \epsilon \in (0,1)$ percent of the rewards (fees and block rewards) if the node commits a fraction $p \in (0,1)$ of the resources committed to verification \cite{pass2017fruitchains}}, communication, or multi-party computation (MPC) bounds, but are much easier to implement and let users directly control their resource usage. Conversely, protocols such as Algorand \cite{micali2016algorand, gilad2017algorand} --- the first camp of protocols that use a small committee of constant or logarithmic size to perform traditional BFT --- rely on public randomness to perform a cryptographic sortition of users. These protocols require much more coordinated communication, but provide stronger fairness guarantees. The final camp of committee selection methods are from protocols that have multiple chains or shards that each require their own set of validators \cite{hanke2018dfinity, kokoris2017omniledger}. In this setting, committee selection corresponds to assigning each of the $n$ participants to $k$ chains in a way that is as close to random as possible. The goal of using randomness in committee selection is to ensure that the protocol is censorship resistant by making it probabilistically impossible for a validator to use an attack, such as a grinding attack, to choose the shard that they get assigned to. In many of these protocols, committee selection is performed repeatedly every $T$ epochs in order to ensure a stronger form of censorship resistance that is insensitive to adaptive adversaries \cite{zamanirapidchain}. 


From a more pragmatic point of view, protocol designers are faced with a variety of trade-offs when choosing a committee selection mechanism. Those concerned with fairness will end up increasing their compute and communication costs and often have to change their Sybil resistance mechanism in response. For example, Algorand requires a common coin (a public randomness beacon for \emph{all} participants, not just the committee) in order to guarantee both fairness and liveness (under asynchrony). This coin, while only used in scenarios under which a rare timeout condition is hit, has a non-trivial cost in terms of ease of implementation and communication complexity. Optimizing these trade-offs ends up tightly coupling the Sybil resistance mechanism to the committee selection mechanism. Intuitively, this suggests that committee selection mechanisms are incomparable, in that their guarantees become tightly coupled to the protocol(s) they are used with.

In this paper, we show that it is possible for two seemingly disparate committee selections mechanisms to be statistically similar. Specifically, we show that Avalanche \cite{avalanche2018}, a protocol that uses multiple rounds of MPC, can be used to generate committees for a social consensus protocol, the Stellar Consensus Protocol \cite{mazieres2015stellar}, that provides BFT guarantees under certain conditions (see \S2.2). Our methods for showing that these two protocols are similar are probabilistic in nature. Avalanche utilizes \emph{private randomness} to allow a participant to randomly sample other participants votes. In order to replicate this mechanism, we outline a generative model for the Avalanche sampling procedure, allowing us to transform Avalanche samples into Stellar's notion of a `socially trusted consensus group', or \emph{quorum slice}. Then, through both numerical estimates and theoretical arguments, we show the existence of a phase transition for when quorum slices from Avalanche satisfy the liveness and safety guarantees of Stellar. Our results suggest that there are more commonalities between committee selection mechanisms than previously thought, and illustrate the importance of numerical simulation for designing committee selection mechanisms. 

\section{Stellar and Avalanche Protocols}

\subsection{Definitions}

We first describe the member preference model used by many Federated Byzantine Agreement (FBA) algorithms.

\begin{definition}
\label{fba_system}
A \emph{Federated Byzantine Agreement System (FBAS)} is a pair $(V, Q)$, where $V$ is a set of participants, and $Q$ is a \emph{quorum slice function} $Q : V \rightarrow 2^{2^V}$, with $v \in S$ for all $S \in Q(v)$. 
\end{definition}

The quorum slice function $Q$ maps each vertex $v$ to a set of one or more \emph{quorum slices}. Intuitively, a quorum slice for node $v$ is a sufficient set of nodes for $v$ to reach agreement. For example, in a traditional Byzantine agreement system \cite{lamport_byzantine}, every node has the quorum slices $Q(v) = \{ \{v\} \cup S : S\subset V, |S| \geq \frac{2}{3}|V| \}$ and as such, has same set of quorum slices as every other node, i.e. $Q(v) = Q(w)$ for all vertices $v, w$.

\begin{definition}
\label{quorum}
A \emph{quorum} is a set of nodes $U \subset V$ such that, for all $v \in U$, there exists some $S \in Q(v)$ with $S \subset U$.
\end{definition}

In other words, a quorum is a set of nodes which can reach agreement without input from any other node in the system. In a traditional Byzantine agreement system with $3n+1$ nodes, any $2n+1$ of the nodes form a quorum.

We can rephrase Definitions \ref{fba_system} and \ref{quorum} in the language of graph theory by appealing to directed hypergraphs \cite{directed_hg}.

\begin{definition}
A FBAS is a directed hypergraph $H=(V, E)$, whose hyperedges are of the form $(v, T)$ for $T \subset V$ and $v \in T$. A \emph{quorum} is a set $U \subset V$ such that, for all $v \in U$, there exists some edge $e = (v, T)$ incident to $v$, with $T \subset U$.
\end{definition}

Two important properties of an FBAS are safety and liveness, which we define below.

\begin{definition}[Safety]
A set of nodes in an FBAS exhibits \emph{safety} if no two nodes externalize different values at the same time.
\end{definition}

\begin{definition}[Liveness]
A node in an FBAS exhibits \emph{liveness} if it can externalize a value without the participation of any adversarial nodes.
\end{definition}

\subsection{Stellar Consensus Protocol}

The Stellar Consensus Protocol (SCP) \cite{mazieres2015stellar} works to preserve BFT guarantees by only requiring consensus among a series of subset of users. In an $n$ user BFT algorithm, one necessarily has $\Omega(n^2)$ messages sent across the network unless routing assumptions can be made. SCP relaxes this by giving each network participant the ability to choose their own quorum slices --- groups of other network participants that they place their trust in. System-wide consensus is then achieved via a modified version of the Paxos algorithm, where each network participant only has to achieve consensus within one of their quorum slices. 

Safety and liveness of the SCP are provably guaranteed when the FBAS hypergraph satisfies certain edge intersection constraints. At the heart of these constraints is the \emph{quorum intersection property} (QIP).

\begin{definition}
\label{qip}
An FBAS has the \emph{quorum intersection property} (QIP) if there do not exist disjoint quorums of the FBAS. That is, $U_1 \cap U_2 \neq \varnothing$ for every pair of quorums $U_1, U_2$.
\end{definition}

To write out the guarantees for SCP, we need the following definition.

\begin{figure*}[h]
    \centering
    \includegraphics[scale=0.4]{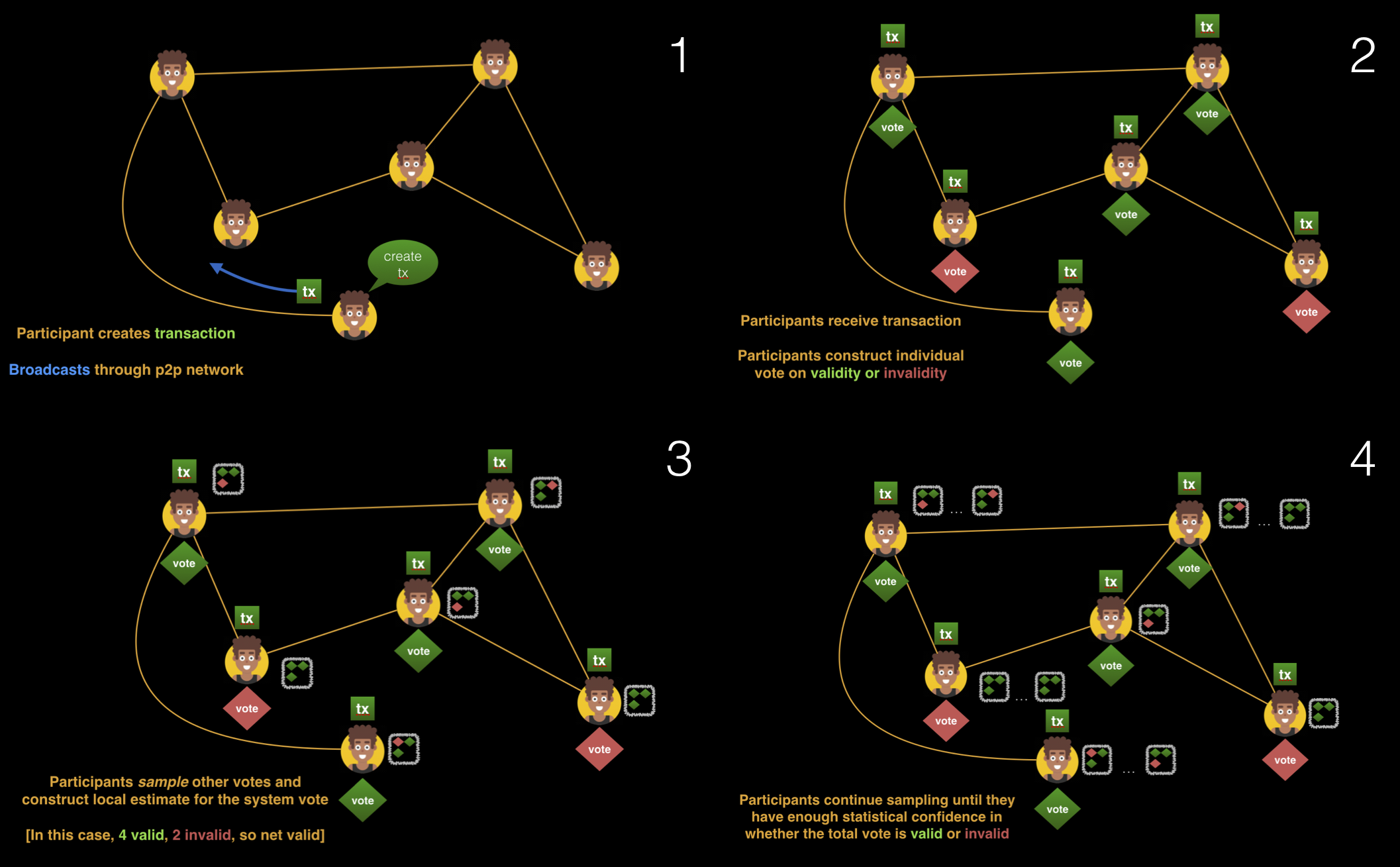}
    \caption{The Avalanche consensus algorithm.(1) A user broadcasts a transaction via a peer-to-peer network to validators. (2) Each validator, upon receipt of this transaction, votes on whether to accept or reject i. (3) Each validator samples votes from other participants in the network using private randomness. (4) Validators repeat this sampling procedure (e.g. sampling with replacement) until they reach an acceptable level of statistical confidence}
    \label{fig:avalanche}
\end{figure*}

\begin{definition}
If $H=(V, E)$ is a FBAS, then for a set of vertices $B \subset V$, we define the  \emph{deletion of $B$ in $H$} as the subgraph $H^B = (V \setminus B, E^B)$, where 
\begin{equation}
    E^B = \{(v, T \setminus B) \mid v \in V\setminus B \text{ and } (v, T) \in E  \}.
\end{equation}
\end{definition}

Formally, the SCP has the following guarantees for safety \cite{mazieres2015stellar}.

\begin{theorem}
\label{stellar_qip}
Let $H$ be an FBAS. Then, the SCP protocol on $H$ exhibits \emph{safety} when $H^S$ exhibits the QIP for any set $S$ of adversarial nodes. 
\end{theorem}

Liveness in SCP is also formally guaranteed through two other technical conditions involving quorum intersection and the topology of $H$.

\subsection{Avalanche}

The Avalanche consensus algorithm \cite{avalanche2018} achieves consensus via repeated probabilistic sampling. Figure \ref{fig:avalanche} depicts how validators interact within Avalanche and solicit repeated votes from other participants. Avalanche provides significantly weaker safety guarantees than the SCP, but is able to achieve a guaranteed communication complexity of $O(k n \log n)$, where $k$ is a parameter that controls the statistical accuracy of a user's local estimate for whether consensus was reached. The algorithm does this by having users use private randomness to sample the set of nodes in the network and decide whether the network has reached consensus or not after taking enough samples. The authors of \cite{avalanche2018} prove, under an assumption of strong network synchrony, a relationship between the number of samples each user takes and the maximum tolerance to Byzantine actors in the system. In this framework, an individual can tune the number of samples they need to reach consensus just as a user can tune their quorum slices in the SCP.

\begin{figure*}
\centering
\includegraphics[scale=0.6]{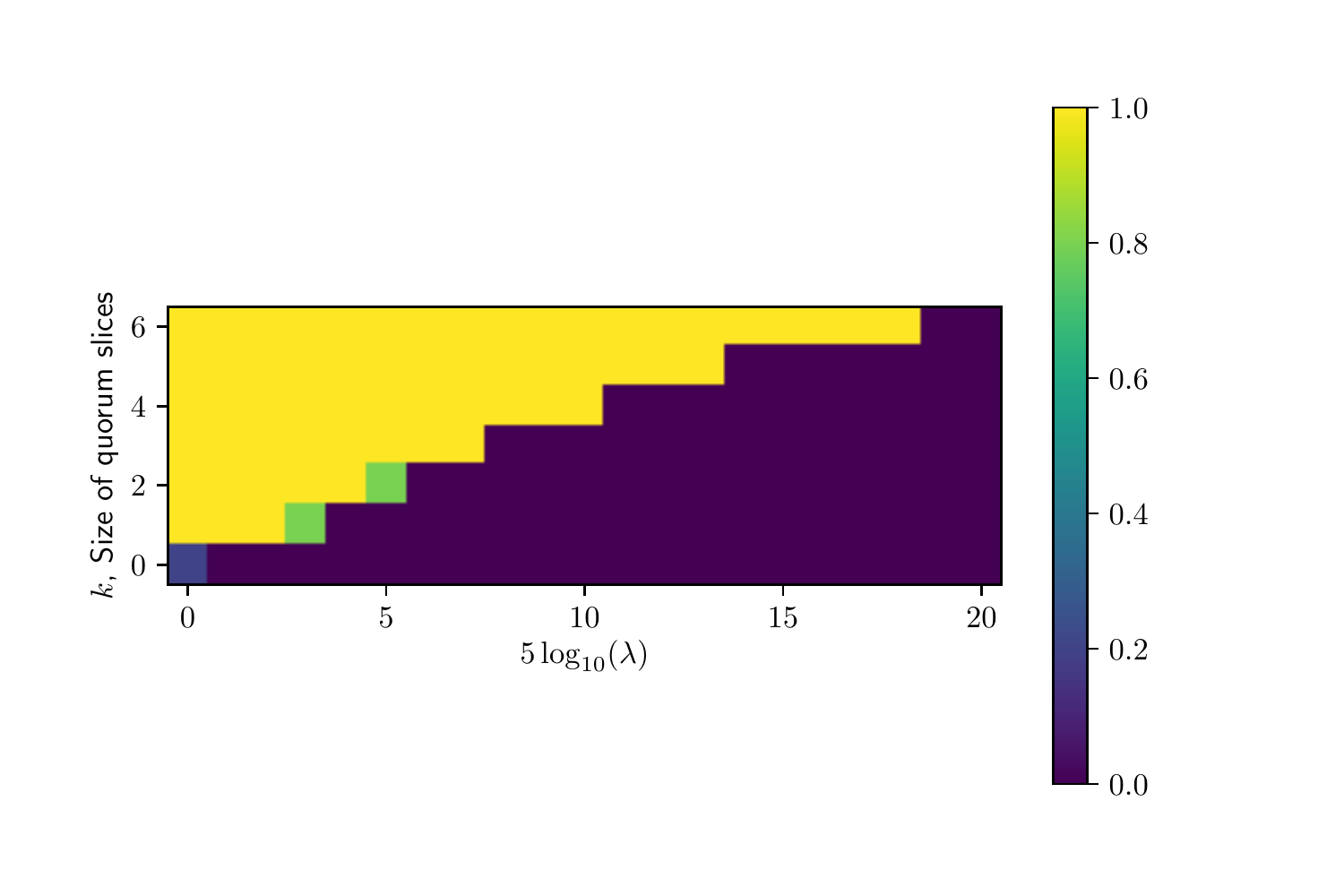}
\caption{Plot depicting when QIP holds for the Avalanche generative model with a network size of $n=16$}
\label{simulation}
\end{figure*} 

The Avalanche algorithm is built up via three previous algorithms: Slush, Snowflake, and Snowball. Slush contains the underlying ideas of the consensus algorithm and is a simplified version of the heavy-hitters algorithm \cite{cormode2005improved, berinde2010space} that is adapted for boolean functions. In Slush and heavy-hitters, the goal is to take samples of size $k$ from a list $L$ of size $N$ and identify the elements $\ell \in L$ such that for some $\phi \geq 0, \; |\{i : L_i =\ell\}| \geq \phi N$. For Slush, we choose $\phi = \frac{1}{2}$. This is done via the following steps:
\begin{enumerate}
\item Sample $k$ elements $\ell_1, \ldots, \ell_k$ from $L$.
\item Compute which indices are most likely to be the majority element of $L$, i.e. $|\{i : L_i =\ell\}| \geq \frac{N}{2}$, based on their frequency in $[\ell_1,\ldots, \ell_k]$.
\item Compute a confidence interval for how likely the most common element seen in $[\ell_1,\ldots, \ell_k]$ is the majority element of $L$.
\item Repeat the above steps (termed a \emph{round}) until the confidence interval shrinks to below a certain threshold, chosen to maximize the probability that the chosen element satisfies the majority condition.
\end{enumerate}
In the case of Slush, the list $L$ consists of the boolean values of whether the $i$th node believes that a certain transaction is valid or not (i.e. the entries of $L$ are either $0$ or $1$). A node samples from this list by using the peer-to-peer network to request votes from different participants. This way, the user doesn't need to know the whole list, but instead treats the samples as results from an online sampling method. Snowflake and Snowball improve upon Slush by reducing the sampling complexity (e.g. the expected number of rounds needed) via an improved confidence calculation, and Avalanche uses Snowball to vote on forks of a directed acyclic graph that represents the ledger. Avalanche provides much weaker guarantees than Stellar by replacing safety (Definition 2.4) with:
\begin{quote}
    P1. Safety. No two correct nodes will accept conflicting transactions
\end{quote}
Note that this new condition is weaker than Definition 2.4 as it provides no bound on how long it will take for two nodes to accept conflicting transactions. This weakened safety requirement, however, is precisely what allows for probabilistic sampling. Avalanche ensures fairness locally in that a participant uses their own randomness source to sample votes from the rest of the network.

\section{Using Avalanche to Generate Quora}

In this section, we show that Avalanche can generate quora that satisfy the Quorum Intersection Property of the Stellar protocol. Thus, Avalanche achieves liveness and safety guarantees similar to SCP (Theorem BLAH).

Our goal in this section is to show that, under certain conditions, Avalanche can generate quora that satisfy the QIP (Definition \ref{qip}), and thus achieve liveness and safety guarantees similar to Stellar (Theorem \ref{stellar_qip}). Using numerical simulations, we empirically observe a phase transition which separates whether or not samples from the Avalanche generative model exhibit QIP with high probability (Figure \ref{simulation}). We also formally prove upper and lower bounds for the phase transition in Theorems \ref{thm_ub} and \ref{thm_lb} respectively, which we illustrate in Figure \ref{phase_transition_diagram}. Our exposition illustrates the importance of numerical simulation in gaining intuition about phase transitions and their value in designing committee selection mechanisms.

\subsection{Generative Model for Avalanche}

We use the following generative model to describe the randomized quorum selection in Avalanche. Label the vertices in our network as $\{1, \dots, n\}$. For each vertex $v \in \{1, \dots, n\}$:
\begin{itemize}
    \item Let $Y_v \sim \Pois(\lambda)$.
    \item Let $A^v_1, \dots, A^v_{Y_v}$ be independently drawn $k-1$ element subsets of $\{1, \dots, n\} \setminus v$.
    \item The \emph{quorum slices} of $v$ are 
    \begin{equation}
        Q_{ava}(v) = \{X^v_1, \dots, X^v_{Y_v}\},
    \end{equation}
    where $X^v_i = A^v_i \cup \{v\}$.
\end{itemize}

Each vertex draws quorum slices from a hypergeometric distribution, matching the analysis done in \cite{avalanche2018}. We also use the standard Poissonization technique \cite{valiant_poisson, jacquet_poisson} to model the number of network samples each node takes. 

\subsection{Simulations}

We create a network of size $n=16$. For $k=2, \dots, 8$ and $\lambda = 10^0, \dots, 10^4$, we draw $10$ samples from our generative model, and test whether those samples satisfy the QIP. Since the general problem of testing whether a network satisfies the QIP is NP-complete \cite{scp_np}, we recursively go through all disjoint pairs of subsets of vertices, checking whether those subsets are quorums.\footnote{We use code from \url{https://github.com/fixxxedpoint/quorum\_intersection} to perform this check efficiently. For $n=16$ it took $40$ hours to run on a MacBook Pro with 8GB RAM.} We exclude $k > n/2$ because any two subsets of size $k$ always intersect. The results of our simulation are shown in Figure \ref{simulation}.

Interestingly, the phase transition in Figure \ref{simulation} formalizes some of the intuition behind how nodes should choose quorum slices. If nodes have ``too many" quorum slices (i.e. exponentially many quorum slices), then disjoint groups of nodes can come to consensus by themselves, and so system-wide consensus cannot be reached. This statistical reasoning improves upon the quorum selection advice of the Stellar paper \cite{mazieres2015stellar}, which states that nodes should ``pick conservative slices that lead to large quorums". 

\subsection{Theory}

\begin{figure*}[h]
\centering
\includegraphics[scale=0.35]{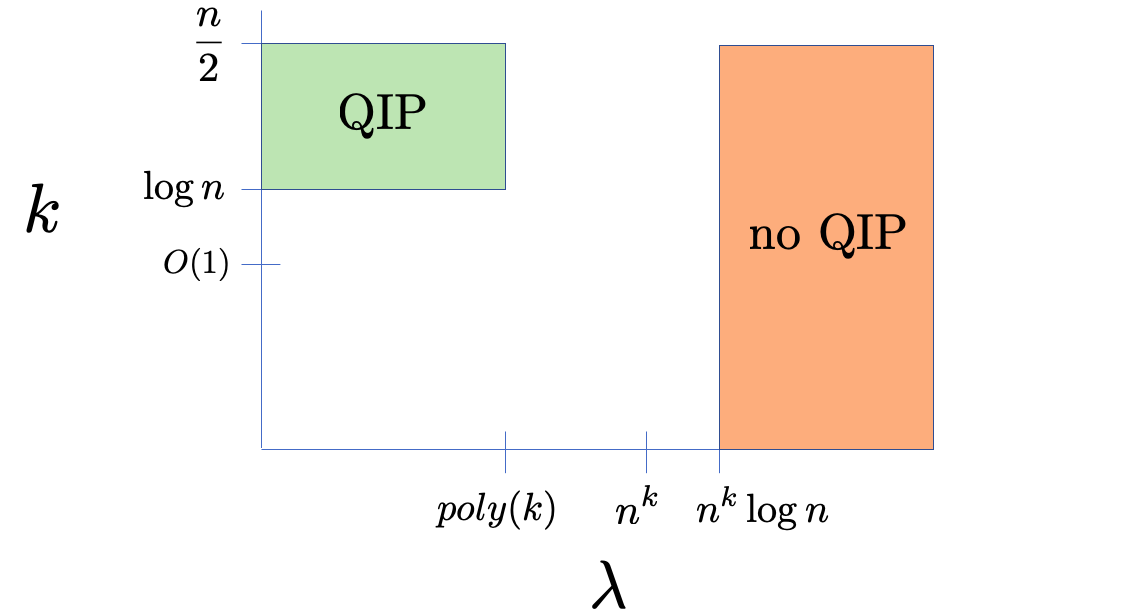}
\caption{Proven lower and upper bounds for when the Avalanche generative model exhibits QIP. Note that QIP is always achieved for $k > n/2$.}
\label{phase_transition_diagram}
\end{figure*} 

From the simulation results in Figure \ref{simulation}, we see a phase transition, where QIP holds for $\lambda < f(k)$, and does not hold for $\lambda > f(k)$. Moreover, the results suggests that $f(k)$ is exponential in $k$. While the exact form of $f$ is unclear, we can prove upper and lower bounds for $f(k)$ (we suspect $f(k) = O(n^k)$). We illustrate our bounds graphically in Figure \ref{phase_transition_diagram}.



\subsubsection{Upper Bound for Phase Transition}

\begin{theorem}
\label{thm_ub}
Assume $k < n/2$. If $\lambda = O(n^k \log{n})$, then $Q_{ava}$ does not exhibit the quorum intersection property with high probability.
\end{theorem}

\begin{proof}
Let $\lambda = 2n^k \log{n}$. We will show that QIP does not hold by showing that $\{1, \dots, k\}$ and $\{k+1, \dots, 2k\}$ are both quora with high probability.

For $v=1, \dots, k$, let $T_v$ be independent geometric random variables with probability parameter $p=\binom{n-1}{k-1}^{-1}$. $T_v$ represents the number of quorum slices vertex $v$ must draw before $\{1, \dots, k\}$ is one of its quorum slices. Thus, $\max(T_1, \dots, T_k)$ is the number of quorum slices vertices $v=1, \dots, k$ must draw for $\{1, \dots, k\}$ to be a quora. 

Let $U_1, \dots, U_k$ be independent exponential random variables, with rate $\alpha = -(\log{1-p})^{-1}$. Since $\lfloor U_i \rfloor$ is geometric with parameter $p$, we have $P(U_i > t) \geq P(T_i > t)$ for any $t > 0$, and so $P(\max(U_1, \dots, U_k) > t) \geq P(\max(T_1, \dots, T_k) > t)$. Setting $t = \frac{\log{n}}{\gamma}$ yields
\begin{equation*}
    \begin{split}
        P(\max(T_1, \dots, T_k) > t) &\leq P(\max(U_1, \dots, U_k) > t) \\
        &= 1 - P(\max(U_1, \dots, U_k) < t) \\
        &= 1-P(U_1 < t)^k \\
        &= 1 - (1-e^{-\log{n}})^k \\
        &= 1-(1-n^{-1})^k \\
        &\leq \frac{k}{n},
    \end{split}
\end{equation*}
where we use Bernoulli's inequality in the last line. Thus, if each vertex in $\{1, \dots, k\}$ draws more than $t= \frac{\log{n}}{\alpha}$ quorum slices, then with high probability, $\{1, \dots, k\}$ is a quorum. Note that
\begin{equation*}
    \frac{\log{n}}{\alpha} = -\frac{\log{n}}{\log{(1-p)}} \leq \frac{\log{n}}{p} \leq n^k \log{n}.
\end{equation*}
Recall that Bennett's inequality \cite{boucheron2013concentration}, when applied to $X \sim \Pois(\lambda)$, yields $P(X \geq \lambda - x) \leq e^{-\frac{x^2}{2\lambda} h(-\frac{x}{\lambda})}$, where $h(x)$ is decreasing in $x$ and bounded above by 2 and $h(-1) = 2$. Thus for $\lambda = 2n^k \log n$ and $x = n^k\log n$, with high probability, $P(Y_v < n^k \log{n}) \leq e^{-2n^k \log(n)}$ for $v = \{1, \dots, k\}$. Thus, with high probability, each vertex will draw at least $\frac{\log{n}}{\alpha}$ quorum slices, and so $\{1, \dots, k\}$ will be a quorum. By similar reasoning, $\{k+1, \dots, 2k\}$ is also a quorum, so QIP does not hold with high probability.
\end{proof}

\subsubsection{Lower Bound for Phase Transition}

\begin{theorem}
\label{thm_lb}
Assume $k = \omega(\log{n})$. If $\lambda = O(n^c)$ for some constant $c$, then $Q_{ava}$ exhibits the quorum intersection property with high probability.
\end{theorem}

Since $O(n^k) = O(n^{\log{n}})$ is super-polynomial, our lower bound states that the phase transition happens after $O(n^k)$ for $k > \log{n}$.

To prove the theorem, we need the following lemma.
\begin{lemma}
\label{lem_lb}
Let $U \subset V$ be a set of vertices with $|U| \geq k$. Under the generative model, the probability that $U$ is a quorum is
\begin{equation}\label{eq:quorum_prob}
    P(U \text{ is a quorum}) = \left( 1 - \exp\left[ -\lambda \frac{(|U|-1)_{k-1}}{(n-1)_{k-1}} \right] \right)^{|U|},
\end{equation}
where $(a)_b$ is the falling Pochhammer symbol.
\end{lemma}

\begin{proof}
We have
\begin{equation*}
\label{lem_1}
    \begin{split}
        P(\text{$U$ is a quorum}) &= \prod_{u \in U} P(\exists S \in Q_{ava}(u), S \subset U) \\
        &= \prod_{u \in U} \big(1 - P(\forall S \in Q_{ava}(u), S \not\subset U)\big) \\
        &= \big(1 - P(\forall S \in Q_{ava}(u), S \not\subset U)\big)^{|U|},
    \end{split}
\end{equation*}
for any fixed $u \in U$. Moreover,
\begin{equation*}
\label{lem_2}
    \begin{split}
        P(\forall S \in Q_{ava}(u), S \not\subset U)\big) &= \prod_{S \in Q_{ava}(u)} P(S \not\subset U) \\
        &= \prod_{S \in Q_{ava}(u)} (1 - P(S \subset U)) \\
        &= \exp\left[ -\lambda \frac{(|U|-1)_{k-1}}{(n-1)_{k-1}} \right],
    \end{split}
\end{equation*}
where in the last step we use that $E[a^Y] = e^{(a-1)\lambda}$ for $Y \sim \Pois(\lambda)$. Putting the above two equations together yields the desired claim.
\end{proof}

Using the lemma, we can now prove our lower bound for the phase transition.

\begin{proof}[Proof of Theorem \ref{thm_lb}]
Note that $\frac{(|U|-1)_{k-1}}{(n-1)_{k-1}} = O(|U|^k / n^k)$. From Lemma \ref{lem_lb}, the expected number of quorums of size $m$ is bounded above by
\begin{equation*}
    \binom{n}{m} (1-e^{-\lambda m^k / n^k})^m \leq n^m \left(2\lambda \frac{m^k}{n^k}\right)^m,
\end{equation*}
where we use that $1-e^{-x} \leq 2x$ for $x > 0$. For $m = o(n)$, the expected number of quorums of size $n$ can be made arbitrarily small. Thus, with high probability, we can assume all quora have linear size.

Next, let $U$ be a quorum of linear size, i.e. $|U| = dn$ for some $d$. Assume $d < \frac{1}{2}$. Write $k = p \log{n}$ for constant $p$, and let $\lambda = n^p$. Then
\begin{equation*}
\begin{split}
    P(\text{$U$ is quorum}) &\leq (1-\exp{(-\lambda d^k)})^{dn} \\
    &= (1-\exp{(-n^p \cdot d^{p \log{n}})})^{dn} \\
    &\leq (1-\exp{(-n^p \cdot \left(1/2 \right)^{p \log{n}})})^{dn} \\
    &= (1-e^{-O(1)})^{dn}.
\end{split}
\end{equation*}

For large $n$, the above probability can be made exponentially small. Thus, with high probability, there are no quorums of size $\leq \frac{1}{2} n$. Since any two sets of size greater than $\frac{1}{2} n$ must have common intersection, it follows that $Q_{ava}$ exhibits QIP with high probability.
\end{proof}


\section{Conclusion}
In this paper, we develop a generative model for the Avalanche consensus protocol, and show that Avalanche can generate FBAS hypergraphs that satisfy the QIP, and therefore achieve safety and liveness guarantees similar to SCP. Because verifying that an FBAS hypergraph satisfies QIP is an NP-complete problem, we argue that one needs simulations and statistical arguments like ours to analyze the SCP in practical settings. More generally, as research into probabilistic consensus mechanisms has increased dramatically since the seminal work of Ben-Or \cite{ben1983another}, there has been an increasing need for a theoretical framework to compare different consensus mechanisms. At the moment, there are a wide variety of proof techniques and guarantees that, at first glance, appear non-comparable. Moreover, the differences in the definitions of what it means to achieve consensus make it hard to discern if a BFT algorithm can be fairly compared to a probabilistic algorithm. However, there are a number of similarities hidden within the analyses of distributed consensus protocols, and this work aims to evince one of these similarities.

In particular, the phase transition that was described was inspired by transitions that are found in statistical physics and integrable probability. For instance, the quorum probabilities of Equation \ref{eq:quorum_prob} are equivalent to cluster size probabilities within the random cluster model under certain limits \cite{grimmett2004random}. This model, which has a variety of sharp phase transitions and a complicated phase diagram, serves as a formal probabilistic tool for analyzing random cluster formations in graphs. Perhaps, by borrowing tools from these fields, it will be possible to show that a larger number of consensus protocols and committee selection algorithms are related by simple statistical transformations. Our hope is that such an inquiry would lead to a simple taxonomy of distributed consensus algorithms and that this paper, which shows that one consensus algorithm can be viewed as a randomized realization of another consensus algorithm, will provide a way to prune this taxonomy.

\bibliographystyle{IEEEtran}
\bibliography{IEEEabrv,stellar_avalanche}

\end{document}